\documentclass[11pt]{article}
\usepackage{fullpage}
\usepackage{amsfonts}
\usepackage{amssymb}
\usepackage{amstext}
\usepackage{amsmath}
\usepackage{xspace}
\usepackage{theorem}
\usepackage{graphicx}
\usepackage{graphics}
\usepackage{colordvi}
\usepackage{colordvi}
\usepackage{multirow}

\newcommand{\IP}{{\sc{Item Pricing}}\xspace}
\renewcommand{\S}{{\mathcal S}}

\newenvironment{proofof}[1]{\noindent{\bf Proof of #1.}}%
        {\hspace*{\fill}$\Box$\par\vspace{4mm}}
\newcommand{\size}[1]{\ensuremath{\left|#1\right|}}
\newcommand{\ceil}[1]{\ensuremath{\left\lceil#1\right\rceil}}
\newcommand{\floor}[1]{\ensuremath{\left\lfloor#1\right\rfloor}}

\newcommand{\opt}{\mbox{\sf OPT}}

\newcommand{\be}{\begin{enumerate}}
\newcommand{\ee}{\end{enumerate}}
\newcommand{\bd}{\begin{description}}
\newcommand{\ed}{\end{description}}
\newcommand{\bi}{\begin{itemize}}
\newcommand{\ei}{\end{itemize}}

\newtheorem{lemma}{Lemma}[section]
\newtheorem{theorem}{Theorem}[section]

\newtheorem{definition}{Definition}[section]
\newenvironment{proof}{\par \smallskip{\bf Proof:}}{\hfill\stopproof}
\def\stopproof{\square}
\def\square{\vbox{\hrule height.2pt\hbox{\vrule width.2pt height5pt \kern5pt
\vrule width.2pt} \hrule height.2pt}}

\renewcommand{\phi}{\varphi}
\newcommand{\eps}{\epsilon}

\newcommand{\E}[1]{\text{\bf E}[#1]}
\newcommand{\pr}[1]{\text{\bf Pr}\left [#1\right]}

\renewcommand{\d}{{\rm d}}

\renewcommand{\u}{\mathcal{U}}

\begin{document}

\begin{titlepage}
\title{Dynamic and Non-Uniform Pricing Strategies
  for Revenue Maximization}
\author{Tanmoy Chakraborty\ \ \ \ \ Zhiyi Huang\ \ \ \ \ Sanjeev Khanna \\
University of Pennsylvania \\
\tt{tanmoy@seas.upenn.edu}, \tt{hzhiyi@seas.upenn.edu},\tt{sanjeev@cis.upenn.edu}
}

\maketitle

\thispagestyle{empty}
\begin{abstract}
\thispagestyle{empty}

We consider the \IP problem for revenue maximization in the
{\em limited supply} setting, where a single seller with $n$ items caters to
$m$ buyers with unknown subadditive valuation functions who arrive
in a sequence.
The seller sets the prices on individual items,
and the price of a bundle of items is the sum of the prices of the
individual items in the bundle. Each buyer buys a subset of yet unsold items
so as to maximize her utility, defined as her valuation of the subset minus the price of the
subset. Our goal is to design pricing strategies,
possibly randomized, that guarantee an expected revenue that is within
a small factor $\alpha$ of the maximum possible {\em social welfare} --
an upper bound on the maximum revenue that can be generated by any pricing
mechanism.

Much of the earlier work
has focused on the unlimited supply setting,
where selling items to some buyer does not affect their availability to
the future buyers. Recently, Balcan
et. al. \cite{BBM08} studied the limited supply setting, giving a
simple randomized algorithm that assigns a single randomly chosen price to all items
({\em uniform pricing strategy}) in the beginning, and never changes it ({\em
  static pricing strategy}). They showed that this strategy guarantees an $2^{O(\sqrt{\log n \log \log
    n})}$-approximation, and moreover, no static uniform pricing strategy can give better than
$2^{\Omega(\log^{1/4} n)}$-approximation.

We relax the space of strategies considered in two directions: we
consider {\em dynamic uniform} strategies, which can change the price
upon the arrival of each buyer but the price on all unsold items is
the same at all times, and {\em static non-uniform strategies}, which
can assign different prices to different items but can never change it
after setting it initially. Dynamic strategies can be especially
useful in online stores, where it is easy to show different prices
to different buyers. We design dynamic and non-uniform pricing strategies
that give a poly-logarithmic approximation to maximum revenue, significantly improving
upon the previous $2^{O(\sqrt{\log n \log \log
    n})}$-approximation. We also give a strengthened lower bound
    of $2^{\Omega(\sqrt{\log n})}$ for approximation factor achieved by any
    static uniform pricing strategy.
Thus in the limited supply setting,  our results highlight a strong separation between the power of dynamic and non-uniform pricing versus
static uniform pricing.
To our knowledge, this is the first non-trivial analysis of
dynamic and non-uniform pricing schemes for revenue maximization.

\end{abstract}

\end{titlepage}

\section{Introduction}

We consider the following \IP problem. Consider a
finite set of items owned by a single seller, who wishes to sell them
to multiple prospective buyers. The seller can price each item
individually, and the price of a set of items is simply the sum of
the prices of the individual items in the set. The buyers arrive
in a sequence, and each buyer has her own {\em valuation
  function} $v(S)$, defined on every subset $S$ of items. We assume that the valuation functions to be {\em
  subadditive}, which means that $v(S) + v(T) \geq v(S\cup
T)$ for any pair of subsets $S, T$ of items. For some results, we shall assume the
valuations to be XOS, that
is, they can be expressed as the maximum of several additive functions.

If a buyer buys a subset $S$ of items $S$, her
{\em utility} is defined as her valuation $v(S)$ of the set minus
the price of the set $S$. Moreover, we assume the {\em limited supply} setting
where a buyer can buy only yet unsold items.
We assume that every buyer is selfish and rational, and
thus always buy a subset of items that maximizes her utility.
The strategy used by the seller in choosing the prices of the items is
allowed to be randomized, and is referred to as a {\em pricing strategy}. The revenue
obtained by the seller is the sum of the amounts paid by each buyer, and our goal is to
design pricing strategies that maximize the expected revenue of the
seller. This problem is made difficult by the fact that the seller has
no knowledge of the valuation functions of the buyers, apart from the promise
that they are subadditive. This is, for instance, in contrast to the Bayesian
mechanism designs for revenue maximization, which assume that the
valuation functions come from a known prior distribution. Optimal
mechanisms, such as that given by Myerson \cite{M81}, exist under this
knowledge.

\medskip
\noindent
{\bf Pricing Strategies:} A {\em uniform} pricing strategy is one where at
any point of time, all unsold items are assigned the same price.
The seller may set prices on the items initially and never change
them, so that cost of an (unsold) item is the same for every buyer. We call such a
strategy to be a
{\em static} pricing strategy. Static pricing is the most widely
applied pricing scheme till date. Alternatively, a seller may set fresh
prices on the arrival of each buyer (without knowing the
buyer's valuation function) -- we shall call this a {\em dynamic}
pricing strategy. Dynamic strategies have become more widely
applicable with the introduction of online stores, since it is quite
easy for online stores to show different prices to different
customers. However, a dynamic strategy in which the price of an item
fluctuates a lot may not be desirable in some applications. So we
introduce an interesting subclass of dynamic strategies, called {\em
  dynamic monotone} pricing strategies, where the price of an item can
only decrease with time.

\medskip
\noindent
{\bf Social Welfare:}
An allocation of items involves distributing the items
among the buyers, and the {\em social welfare} of an allocation the
sum of the buyers' valuations for the items received by each of
them. We denote the maximum social welfare, achieved by any
allocation, by $\opt$.
We measure the performance of a pricing strategy as the
ratio of the maximum social welfare against the smallest expected
revenue of the strategy, for any adversarially chosen ordering of the
buyers. (Some of our results, where it will be explicitly stated,
shall consider expected revenue under the assumption that the order in
which buyers arrive is uniformly random.) If
this ratio is at most $\alpha$ on any instance (where $\alpha$ can
depend on the size of the instance), we say that the strategy achieves
an $\alpha$-approximation. Note that the maximum social welfare is an
upper bound on the revenue the seller can obtain under any
circumstance. In fact, there
exists simple instances with $n$ items and a single buyer where the
maximum social welfare is $\log n$, but the revenue can never exceed
$1$ for any pricing function \cite{BBM08}. Thus we are comparing the
performance of our strategies against a bar that is significantly
higher than the optimal strategy, and we can never hope to achieve
anything better than a logarithmic approximation. Our general goal
is to design pricing strategies that achieve polylogarithmic
approximation.

\medskip
\noindent
{\bf Related Work:}
The \IP problem is closely related to the extensive body of
literature in combinatorial auctions \cite{CSS05}, which is the
setting as described above, except that the buyers need not be
arriving in a sequence but instead may place simultaneous bids on the
items. A lot of recent literature has focused on social welfare maximization.
This includes efficient approximation algorithms for computing maximum
social welfare given oracle access to the valuation functions
(eg. \cite{F06}), as well as on efficiently computable mechanisms that
maximize social welfare and are truthful (eg. \cite{LOS02,
  LS05, DNS06, D07}). For the first problem, Feige \cite{F06} gave a
constant approximation for subadditive buyers, while for the second
problem, Dobzinski et. al. \cite{DNS06,D07} gave logarithmic approximation
 when buyers have XOS valuations and subadditive valuations
 respectively. The mechanism achieving this approximation is in fact a
 static uniform pricing strategy.

A fair amount of research has focused on algorithms and truthful
mechanisms for revenue maximization as well, but it has mostly considered
the {\em unlimited supply setting} \cite{HK07}, where
unlimited number of copies of each item is available to
the seller. So one buyer receiving an item does not stop another
buyer receiving the same item. Thus, the order in which buyers arrive
has no effect on the performance of the mechanism, and in fact, the
buyers can be handled independently.
Some research has been directed towards developing new
truthful mechanisms that maximize revenue \cite{BBHM05, FGHK02,
  GHW01}, while others have focused on designing strategies for item
pricing that maximizes revenue. The item pricing
problem has received special attention because it is and has been
the most widely applied mechanism for a seller wishing to sell items
to potential buyers. All the research has focused only on
designing static strategies (eg. \cite{GHKKKM05, BB06, AFMZ04, BK06,
  DFHS06}), and moreover, some of them have restricted their
attention to finding envy-free pricing, which implies that the buyers
come simultaneously, and the pricing must ensure that two buyers does
not seek the same item. Moreover, most of these works assume severely
restricted classes of valuation functions. For example,
\cite{GHKKKM05,BB06} assume that all buyers are single-minded
bidders.
Their strategies were not only static but also
uniform. Unsurprisingly, finding envy-free pricing is hard
\cite{DFHS06}, and their results do not extend to more general buyer
valuations such as XOS or subadditive. In all this work, the
performance of a strategy has been measured as the ratio of the
 maximum social welfare to the expected revenue obtained.

More recently, Balcan, Blum and Mansour \cite{BBM08} considered static
pricing strategies with the objective of revenue maximization, in the
limited supply setting, with subadditive buyer
valuations. In the unlimited supply setting, they designed a pricing strategy
that achieve revenue which is logarithmic approximation to the maximum
social welfare even for general valuations. The strategy, again, was a
uniform strategy. This result was also
proved independently in \cite{BHK08}. However, in the limited supply
setting, they could only get a $2^{O(\sqrt{\log n \log \log n})}$ factor
approximation using a static uniform strategy. Crucially, they ruled
out the existence of static uniform strategies that achieve anything
better than a $2^{\Omega((\log n)^{1/4})}$ approximation, even if the
buyer valuations are XOS, and the ordering of buyers is assumed to be
chosen uniformly at random. Thus their result distinguished the
limited and unlimited supply settings. This impossibility of getting a
good (polylogarithmic) approximation is a consequence of being
restricted to static uniform strategies, and it remains impossible
even if the seller knew the buyer valuations, and had unlimited
computational power. Further, almost all mechanisms in these
related problems have only used a single price for all items. It is,
therefore, natural to consider dropping one of these restrictions,
namely, look at dynamic uniform strategies and static non-uniform
strategies, both of which use multiple prices, and attempt to find
better guarantees on the revenue.

\paragraph{Our Results and Techniques}

The table below summarizes our results on the \IP problem in
the limited supply setting, along with relevant earlier work. Our contributions are labeled with the
relevant theorem numbers.

\begin{table}
\centering
\begin{minipage}{\textwidth}
\begin{tabular}{| p{2.9cm} || p{2.5cm} | p{2.9cm} || p{2.6cm} | p{3cm} |}
\hline
\multirow{2}{3.2cm}{{\bf Type of Pricing Strategy}} &
\multicolumn{2}{c||}{{\bf Subadditive buyer
  valuations}} & \multicolumn{2}{c|}{{\bf$\ell$-XOS buyer valuations}}\\
\cline{2-5}
  & {\bf Algorithm \footnotemark[1]}& {\bf Lower Bound
  \footnotemark[2]} & {\bf Algorithm \footnotemark[1]} &
{\bf Lower Bound \footnotemark[2]}\\
\hline
Dynamic Uniform Pricing & $O(\log^2 n)$ [Thm.
\ref{thm-dynamic}]  \footnotemark[3] & $\Omega\left(\left(\frac{\log n}{\log\log n}\right)^2\right)$
 [Thm. \ref{dyn-hardness}] & $O(\log^2 n)$
 [Thm. \ref{thm-dynamic}]  \footnotemark[3] & $\Omega\left(\left(\frac{\log
   n}{\log\log n}\right)^2\right)$ \ \ \ \ \ \ \ \ \ \ \ \ \ \ \ \ \
\ \ \ \ \ \ \ \ \ [Thm. \ref{dyn-hardness}] \\
\hline
Dynamic Monotone Uniform Pricing & $O(\log^2 n)$ [Thm. \ref{thm-dynmon}]
\footnotemark[4] \footnotemark[5]  & $\Omega\left(\left(\frac{\log n}{\log\log n}\right)^2\right)$
[Thm. \ref{dyn-hardness}] & $O(\log^2 n)$ [Thm.
\ref{thm-dynmon}]  \footnotemark[4] \footnotemark[5]  &
$\Omega\left(\left(\frac{\log n}{\log\log n}\right)^2\right)$ \ \ \ \
\ \ \ \ \ \ \ \ \ \ \ \ \ \ \ \ \
[Thm. \ref{dyn-hardness}] \\
\hline
Static Uniform  Pricing & $2^{O(\sqrt{\log n\log \log n})}$ [BBM08
\cite{BBM08}] \footnotemark[3]& $2^{\Omega({\sqrt{\log n}})}$
\ \ \ \ \ \ \ \ \ \ \ \ \ \ \ \ \ \ \ \ \ \ \ \ \
[Thm. \ref{stat-unif-lb}] & $2^{O(\sqrt{\log n\log \log n})}$
[BBM08 \cite{BBM08}] \footnotemark[3]& $2^{\Omega({\sqrt{\log n}})}$
\ \ \ \ \ \ \ \ \ \ \ \ \ \ \ \ \ \ \ \ \
[Thm. \ref{stat-unif-lb}] \\
\hline
Static Non-uniform Pricing & $2^{O(\sqrt{\log n\log \log n})}$ [BBM08
\cite{BBM08}] \footnotemark[3]& $\Omega\left(\log n\right)$ [BBM08
\cite{BBM08}] & $O\left(m\log l \log^3 n\right)$
[Thm. \ref{thm-non-uniform}] \footnotemark[3]& $\Omega\left(\log
  n\right)$
\ \ \ \ \ \ \ \ \ \ \ \ \ \ \ \ \ \ \ \ \ \ \ \ \ \ \ \ \ \ \ \
[BBM08 \cite{BBM08}]\\
\hline
\end{tabular}
\footnotetext[1]{All algorithms assume that the seller knows $\opt$ up
  to a constant factor. This assumption can be removed by worsening
  the approximation ratio by a factor of $\log \opt (\log\log \opt)^2$.}
\footnotetext[2]{All lower bounds are in the full information setting,
where the seller knows the buyers' valuations, the number and arrival
order of buyers, has unbounded computational power, and can even force
the arrival order of buyers!}
\footnotetext[3]{Buyers arrive in an adversarial order. Thus the
  algorithm satisfies the upper bound for any order of buyers,
  including the order that minimizes expected revenue.}
\footnotetext[4]{Buyers arrive in a uniform random order, that is, every
permutation of buyers is equally likely. The bound is on the expected
revenue under this assumption.}
\footnotetext[5]{This algorithm also assumes that the seller knows the
  number of buyers $m$ up to a constant factor, and it is
  deterministic. The assumption can be removed by making it
  randomized, and worsening the approximation ratio by a factor of
  $\log m(\log\log m)^2$.}
\end{minipage}
\label{summary-table}

\end{table}

We strengthen the hardness result of Balcan
et. al. \cite{BBM08} by constructing instances with XOS valuations
where uniform pricing functions cannot achieve any better than a
$2^{\Omega(\sqrt{\log n})}$ approximation, even if the seller knew the
buyer valuations, and had unlimited computational power. Our basic
construction is essentially similar to one given in \cite{BBM08}. We further extend our
construction so that all buyers have the same XOS valuation function,
so that the revenue is small for any order of buyers.
Alternatively, we can extend it so that every buyer has an XOS buyer
valuation (possibly different from the other buyers) that can be
expressed as the maximum of three additive functions.

In contrast, we design a simple randomized dynamic
uniform pricing strategy such that its expected revenue is $O(\log^2 n)$
approximation of the optimal social welfare, when the valuation
functions are subadditive. The strategy randomly chooses a
threshold at the beginning, and then in each round, randomly chooses a
price above the threshold (but less than $\opt$) and puts this price
on each unsold item. By using a fresh random price (from a suitable
set of prices) in each round, we guarantee, in expectation, to
collect a large fraction of the revenue that can be obtained in that
round from the remaining items.

The dynamic uniform pricing strategy described above achieves a high revenue, but
requires random fluctuation in the price of an unsold item. This may not be a desirable property in some
applications.
We design a dynamic monotone uniform pricing strategy where the price of any unsold items
only decreases over time. We show that if the ordering of buyers is assumed to
be uniformly random, that is, all permutations of buyers are equally
likely, then the
expected revenue is an $O(\log^2 n)$ approximation of
the optimal social welfare. The strategy is in fact deterministic
provided the seller knows estimates of $\opt$ and $m$ up to a
constant factor. Deterministic strategies giving good approximation in
such limited information settings are rare. We emphasize here that our
lower bound for static uniform pricing holds for any ordering of buyers.

We show that the performance of our dynamic uniform pricing strategies are
almost optimal among all dynamic uniform strategies, by showing that
even if the seller knew the buyer valuations, had unlimited
computational power, and could even force a particular ordering of the
buyers, there exists instances with XOS valuations where the
seller can achieve a revenue of at most $\opt (\log \log n)^2/\log^2
n$ if she is restricted to choosing a uniform dynamic pricing.

All our algorithms as well as the algorithms in \cite{BBM08} assume
that $\opt$ is known to the seller up to a constant factor. Moreover,
our dynamic monotone strategy assumes that the number of buyers $m$ is
known to the seller up to a constant factor. As Balcan
et. al. \cite{BBM08} pointed out, for any parameter that is assumed to
be known up to a constant factor, if the seller instead knows an upper
bound of $H$ and a lower bound of $L$ on
the optimum, then this assumption can be removed by guessing $\opt$
with a suitable distribution, worsening the approximation ratio by a
factor of $\Theta(H/L)$. If the seller instead knows that $\opt\geq
1$, but knows no upper bound, then the assumption can be removed by
worsening the approximation ratio by a factor of $\Theta(\log x (\log
\log x)^2)$ in the approximation, where $x$ is the said parameter.

Finally, we give a static non-uniform strategy that gives an $O(m\log
\ell \log^3 n)$-approximation if the buyers' valuations are XOS
valuations that can be expressed as the maximum of $\ell$ additive
components. Note that when the order of buyers is adversarial, the
hard instance for static uniform pricing has only two buyers, and
their valuation functions are the maximum of $o(\log n)$ additive
functions components. In particular, our strategy gives
polylogarithmic (in $n$) approximation when the number of buyers are
small (polylogarithmic in $n$), and has XOS valuations which are the
maximum of quasi-polynomial (in $n$) additive components. It is worth
noting that our lower bound for dynamic uniform strategies also
satisfies these properties.


\section{Preliminaries}\label{prelim}

In the \IP problem, we are given a single seller with a set $I$ of $n$ items
that she wishes to sell. There are $m$ buyers, each with their own {\em valuation
  function} defined on all subsets of $I$. A buyer with valuation
function $v$ values a subset of items $S\subseteq I$ at $v(S)$.
The buyers arrive in a sequence, and each buyer visits the
seller exactly once. The seller is allowed to
set a price on each item, and the price of a subset of items is the
sum of the prices of items in that subset. For every item sold to the
buyers, the seller receives the price of that item. Note that an item
can be sold at most once. So a seller can only offer those items to a
buyer that has not been sold to any previous buyer. The {\em revenue}
obtained by the seller is the sum of the prices of all the sold items.

Each buyer buys a subset of the items shown to her that maximizes her
{\em utility}, which is defined as the value of the subset minus
the price of the subset. This is clearly the behavior that is most
beneficial to the buyer. The \IP problem is to design
(possibly randomized) pricing strategies for the seller that maximizes
the expected revenue of the seller.

Unless noted otherwise, all our algorithmic results
will assume that the seller has no knowledge of the order of arrival of the buyers,
total number of buyers, or the valuation functions of buyers.
We refer to a setting as the {\em full-information setting}
if all these parameters are known to the seller.

\medskip
\noindent
{\bf Valuation Functions:}
Throughout this paper,  we will assume that the buyer valuation function $v$ is {\em
  subadditive}, which means that $v(S) + v(T) \geq v(S\cup T)\ \forall
S\subseteq I, T\subseteq I$.
{\em Unless explicitly stated otherwise, this will be the only assumption on the buyer
valuation functions.}
For some results, we shall assume the buyer valuations to be more
restrictive than subadditive.

\begin{definition}
A subadditive valuation function $v$ is called an {\em XOS valuation} if it can be
expressed  as $v(S) = \max \{a_1(S),a_2(S)\ldots a_{\ell}(S)\}\ \forall
S\subseteq I$ on all subsets of items $S$, where $a_1,a_2\ldots a_{\ell}$
are non-negative additive functions. The functions $a_1, a_2, ..., a_{\ell}$
are referred to as
the {\em additive valuation components} of the XOS valuation $v$.
We say that $v$ is an $\ell$-XOS function if it can be expressed
using at most $\ell$ additive valuation components.
\end{definition}

We note that a $1$-XOS function is simply an additive function, that
all XOS valuations are subadditive, and that
not all subadditive valuations can be expressed as XOS valuations.

\medskip
\noindent
{\bf Pricing Strategies:} We will study the power of some natural classes of
pricing strategies.

\begin{definition}
A pricing strategy is said to be {\em static} if the seller initially
sets prices on all items, and never changes the prices in the
future. A pricing strategy is said to be {\em dynamic} if the seller
is allowed to change prices at any point in time.
A dynamic pricing strategy is also said to be {\em monotone} if the
price of every item is non-increasing over time.
\end{definition}

\begin{definition}
A pricing strategy is said to be {\em uniform} if at all points in
time, all unsold items are assigned the same price.
\end{definition}

\subsection{Notation}

For a buyer $B$ with a valuation function $v$,
we use $\Phi(B,I,p)$ to denote a set of items that the buyer $B$ may
buy when presented with set $I$ of items, each of which are priced at $p$.
Since $v(S) - p|S|$ is the utility if the buyer buys
the set $S$, so $\Phi(B,I,p)=\text{argmax}_{S\subseteq J} v(S) -
p|S|$ maximizes the utility. Note that there may be multiple
possible sets that maximize the utility. In this paper, when we
make a statement involving $\Phi(B,I,p)$, the statement shall hold for
any choice of these sets.
We shall denote the maximum utility as
$\u(B,I,p)$; note that in contrast to $\Phi(B,I,p)$, the value $\u(B,I,p)$
is uniquely defined. When the underlying buyer $B$ is clear from the context, we shall
denote these two values as $\Phi(I,p)$ and $\u(I,p)$
respectively. Moreover, if the set of available items $I$ is also
clear from the context, then we shall denote these two values as
$\Phi(p)$ and $\u(p)$ respectively. For any set $S$ and a buyer with
valuation function $v$, we define $H_v(S)=\max_{S'\subseteq S} v(S')$
as the maximum utility the buyer can get if all items in $S$ are
offered to her at zero price.

\begin{definition}
We say that a set of items $S$ is {\em supported at a price $p$}
with respect to some buyer $B$ with valuation function $v$, if $B$
buys the entire set $S$ when the set $S$ is presented to $B$ at a
uniform pricing of $p$ on each item.
\end{definition}

The following lemma follows easily from the fact that the valuation
functions are subadditive, and was proved by Balcan
et. al. \cite{BBM08}. 
\begin{lemma}\label{unsold-value}
Let $S$ be a set of items that is supported at price $p$. with respect
to a buyer $B$ with valuation function $v$. Then $v(S')\geq p|S'|$ for
all $S'\subseteq S$.
\end{lemma}
\begin{proof}
Suppose not. Then there exists $S' \subset S$
with $v(S') < p|S'|$. By subadditivity of $v$, we know
$v(S') + v(S\setminus S') \ge v(S)$, and hence $v(S\setminus S') \ge v(S) - v(S')$.
Then the utility for $B$ of buying the set $(S \setminus S')$ is
at least $v(S) - v(S') - p|S/S'|$. But

$$\left( v(S) - v(S') - p|S/S'| \right) - p|S'| + p|S'| = \left( v(S) - p|S| \right) - v(S') + p|S'| > v(S) - p|S|,$$

contradicting the assumption that buyer picks set $S$ at price $p$.
\end{proof}

\subsection{Optimal Social Welfare and Revenue Approximation}

We now define optimal social welfare, the measure against which we
evaluate the performance of our pricing strategies.

\begin{definition}
An {\em allocation} of a set $S$ of items to buyers $B_1,B_2\ldots B_m$ with
valuations $v_1,v_2\ldots v_m$, respectively, is an $m$-tuple
$(T_1,T_2, \ldots ,T_m)$ such that $T_i \subseteq S$ for $1 \leq i \leq m$,
and $T_i \cap T_j =\emptyset$ for $1\leq i, j\leq m$.
The {\em social welfare} of an allocation is
defined as $\sum_{i=1}^m v_i(T_i)$, and an allocation is
said to be a {\em social welfare maximizing allocation}
if it maximizes $\sum_{i=1}^m v_i(T_i)$.
The {\em optimal social welfare}
$\opt$ is defined as the social welfare of a social welfare maximizing
allocation.
\end{definition}

Clearly, $\opt$ is an upper bound on the revenue that any
pricing strategy can get. Let $R$ be the revenue obtained by the
strategy, which is the sum of the amounts paid by all the buyers.

\begin{definition}
A pricing strategy is said to achieve an $\alpha$-approximation if the
expected revenue of the strategy $\E{R}$ is at least $\opt/\alpha$.
\end{definition}

Unless stated otherwise, the expected revenue is computed with
adversarial ordering of the buyers, that is, the ordering that
minimizes the expected revenue of the strategy. In other words, we
require a strategy to work well irrespective of the order of buyers in
which they arrive.

Note that $\opt$ is not a tight upper bound on the maximum revenue
that can be achieved by any pricing strategy, even with full knowledge
of buyer valuations and unbounded computational power. In fact, the
following example was given in Balcan et. al. \cite{BBM08}: if there
is a single buyer with valuation function $v(S)=\sum_{i=1}^{|S|} 1/i$,
then for any pricing of the $n$ items, the revenue is at most $1$,
while $\opt=\Theta(\log n)$. This shows that nothing better than a
logarithmic approximation can be achieved in the absence of any other
assumption on the buyer valuations.

\subsection{The Single Buyer Setting with Uniform Pricing Strategies}\label{borrowed-lemmas}

Balcan et. al. \cite{BBM08} considered the setting where there is an
unlimited supply of each item, so that no buyer is affected by items
bought before her arrival. In particular, if there is only a single
buyer, then there is no distinction between limited and unlimited
supply, as long as the buyer never wants more than one copy of the
same item. For the single buyer case, Balcan et. al. \cite{BBM08} gave
an $O(\log n)$ approximation, and in the process proved some lemmas
that will be useful for our algorithmic results in the limited supply
setting as well.

Suppose a set $S$ is being shown to a buyer $B$ with valuation
function $v$. The optimal social welfare in this single buyer
instance is $H_v(S)$. We consider setting a uniform price, that is,
the same price on all items. The following lemma states that the
number of items bought monotonically decreases as the price on the
items is increased. It was proved by Balcan et. al. \cite{BBM08}.

\begin{lemma}\label{decreasing-size}(Lemma 6 of \cite{BBM08})
Suppose a buyer $B$ is offered a set $S$ of items using a uniform
pricing.
Then for any $p > p' \ge 0$, if $B$ buys
$\Phi(p)$ if all items are priced at $p$, and $\Phi(p')$ if all items
are priced at $p'$, then $|\Phi(p)|\leq |\Phi(p')|$.
Thus there
exist prices $\infty=q_0>q_1>\ldots>q_l>q_{l+1}=0$ and
integers $0=n_0<n_1<\ldots<n_l\leq |S|$ such that when items in $S$ are
uniformly priced at $p \in [q_{t+1}, q_t)$ there is a subset $S'
\subseteq S$ of size $n_t$ that is supported at price $p$, and
the utility $\u(p)$ of the buyer $B$ satisfies
\begin{equation} \label{netutil-eqn}
 \u(p)=\u(q_{t})+n_t(q_{t}-p)
\end{equation}
\end{lemma}

Since the empty set maximizes utility when the price is $q_1$, we
get that $\u(q_1)=0$. Moreover, the utility at price $q_{l+1}=0$
is $\u(q_{l+1})=H_v(S)$. Thus we get that $H_v(S)=\sum_{t=1}^l n_t
(q_t - q_{t+1})$.

The following lemma is a slight variation of Lemma 8 of
\cite{BBM08}.

\begin{lemma}\label{BBM-revsum}
Suppose a set $S$ is being shown to a buyer $B$, with valuation
function $v$, using a uniform price. Let $H'$ be any number such that
$H'\geq H_v(S)$.  Let $\gamma >1$, and let $p[t]=H'/\gamma^t$. Then,
for any $k\geq 0$, we have $\sum_{t=1}^k p[t]|\Phi(p[t])| \geq
\frac{1}{\gamma-1} \left( H_v(S) - \frac{|S|H'}{\gamma^k} \right)$.
\end{lemma}
\begin{proof}
Since $H_v(S)= \sum_{t=1}^l n_t(q_t - q_{t+1})$, it can
be seen as an
integral of the following step function $f$ from $q_{l+1}=0$ to $q_1$:
in the range $[q_{t+1},q_t)$, $f$ takes the value $n_t$. So we can upper
bound $H_v(S)$ by an upper integral of $f$. Note that $|f(p)\leq
\Phi(p)|\leq |S|$, and also that $f$ is a decreasing function. Since
$p[0] =H' \geq H \geq q_1$, we get
\begin{equation*}
\begin{aligned}
H_v(S) &= \sum_{t=1}^l n_t(q_t - q_{t+1})  = \int_0^{q_1} f(x)\d x\\
& \leq \int_0^{p[k]} f(x)\d x + \sum_{t=0}^{k-1}
(p[t]-p[t+1])f(p[t+1])\\
& \leq \int_0^{p[k]} |S|\d x + \sum_{t=0}^{k-1} (\gamma - 1)
p[t+1]|\Phi(p[t+1])|\\
&= (\gamma - 1) \sum_{t=1}^k p[t]|\Phi(p[t])| + |S|p[k]
\end{aligned}
\end{equation*}
Thus we get that $\sum_{t=1}^k p[t]|\Phi(p[t])| = \frac{1}{\gamma-1}
\left( H_v(S) - |S|H'/\gamma^k \right)$.
\end{proof}

Briefly, Lemma \ref{BBM-revsum} will be used as follows: if one of
$\{H',H'/2,H'/4 \ldots H'/2^k\}$ is chosen uniformly at random and set
as the uniform price for all items in $S$, then for a sufficiently
large choice of $k$, the revenue obtained is $\Omega(H_v(S)/k)$. This
will happen when the right-hand-side of the equation in the lemma
evaluates to $\Omega(H_v(S))$. We shall frequently use this lemma, and
with $H'=\Theta(H)$, our choice of $k$ will be logarithmic in the
number of items.

\subsection{Optimizing with Unknown Parameters}

Almost all our algorithms use the following lemma, which was implicitly
mentioned in the Appendix of \cite{BBM08}. It tells us that
strategies can be allowed to assume that it approximately knows the
value of some parameters, as long as the parameters are not too large,
since these assumptions can be removed by guessing the value of these
parameters and getting it correct with inverse-polylogarithmic
probability. The lemma below is applicable to the \IP problem
with multiple buyers, and to both static and dynamic pricing
strategies.

\begin{lemma}\label{param-assumption}
Consider a pricing strategy $\S$ that gives an $\alpha$-approximation in
expected revenue, provided the seller knows the
value of some parameter $x$ to within a factor of $2$. Then if the seller
instead only knows that $L\leq x < H$, there exists
a pricing strategy $\S'$ that gives an $O(\alpha
\log (H/L))$ approximation in expected revenue,
where $L$ and $H$ are powers of $2$. If the seller instead
only knows that $x\geq 1$ but no upper
bound, then for any constant $\eps > 0$, there exists a pricing strategy $\S''$
that gives an $O(\alpha \log x (\log \log x)^{1+\eps})$
approximation in expected revenue.
\end{lemma}
\begin{proof}
We construct a pricing strategy by approximately guessing the
value of $x$, up to the nearest power of $2$, using a suitable
distribution, at the beginning, and using this estimate in the given
pricing strategy. Our revenue is assured only when our guess is
correct, and we count only that revenue in our analysis.

In the case where $L$ and $H$ are given, we guess $x$ from the set
$\{L,2L,4L\ldots H\}$, so that our guess of $x$ is correct within a
factor of $2$, with probability
at least $\Omega(1/\log (H/L))$. Since the given pricing strategy
gives an expected revenue of $\Omega(\opt/\alpha)$ when the guess is
correct within a factor of $2$, we get an expected revenue of at least
$\Omega(\opt/\alpha \log (H/L))$.

In the second case, where the seller only knows that $x\geq 1$, then
we guess that $x=2^i$ with probability $\frac{1}{c (i\log^{1+\eps} i)}$, where
$c=\sum_{i=1}^{\infty} 1/i\log^{1+\eps} i$, which is finite. If $x$ is
between $2^i$ and $2^{i+1}$, then the probability
of guessing $x$ correctly within a factor of $2$ is at least
$\Omega(1/i \log^{1+\eps} i) = \Omega(1/(\log x(\log\log
x)^{1+\eps}))$, so the expected revenue is at least
$\Omega(\opt/(\alpha \log x (\log\log x)^{1+\eps}))$.
\end{proof}

\section{Improved Lower Bounds for Static Uniform Pricing}\label{stat-lb-section}
We show some lower bounds for static uniform
pricing. The core of our construction is the same as the lower bound
construction in Balcan et. al. \cite{BBM08}, but with improved
parameters, and our lower bound almost matches the upper bound in
\cite{BBM08}. We are also able to strengthen our construction to the
case of identical buyers as well as to the case where each buyer uses
simple XOS functions with only 3 additive components.
The following theorem summarizes our lower bound results about static
uniform pricing. 

\begin{theorem}\label{stat-unif-lb}
There exists a set of buyers with XOS valuations, such that if the
seller is restricted to a static uniform pricing strategy, then
even in the full information setting, for any choice of price,
the revenue produced is at most
$\opt/2^{\Omega(\sqrt{\log n})}$, where $n$ is the number of
items. Additionally, one of the following (but not both) can also be ensured, with the revenue still being at most
$\opt/2^{\Omega(\sqrt{\log n})}$:
\begin{itemize}
\item The valuations of all the buyers can be expressed as 3-XOS
	functions.
\item All buyers have identical valuation function.
\end{itemize}
\end{theorem}

We now present the proof of Theorem \ref{stat-unif-lb}. We first
construct an instance with two buyers whose valuations
consist of only three additive components each, such that if buyer 1
arrives before buyer 2, then the revenue obtained will satisfy the
required upper bound. This part of our construction is very similar to
that given in \cite{BBM08}, with some changes in the parameter that
allows us to improve the lower bound result from $2^{\log^{1/4} n}$ to
$2^{\Omega(\sqrt{\log n})}$. We shall then extend this construction to
instances where no ordering of buyers gives a high revenue, where all
buyers have identical valuations and finally where all buyers have
3-XOS valuations.

\medskip

\paragraph{A hard two-player instance:}
Let $X > 1$ and $Y < 1$ be two parameters that shall be fixed later.
Consider an instance of the problem as described below. Let $n_o$ be
a positive parameter, and as we shall see, the number of items will
be between $n_0$ and $2n_0$. There are two buyers with buyer
valuations $v_1$ and $v_2$. Let $S_0,S_1,\cdots,S_{6k+2}$ be a
partition of items, where $k=\floor{\sqrt{\log n_0}/3}$. There is
a subset $S'_i\subseteq S_i$ of items of high valuation for each $i$.
$|S_i| = n_0/X^i $ and $|S'_i| =
n_0/X^{i+1} $. Buyer 1 does not value the items in $S_i\setminus
S_i'$, that is, $v_1(S_i\setminus S'_i) = 0$. Buyer 1 values the
items in $S_i'$ equally, such that $v_1(S'_i) = cY^i$.
Similarly, buyer 2 values the items in $S_i'$ equally, such that
$v_2(S'_i) = Y^i-Y^{i+1}$, and values the items in $S_i\setminus S'_i$
equally such that $v_2(S_i\setminus S'_i) = Y^{i+1}$. Finally, the
valuation function of each buyer consists of {\em three additive
components} which are additive inside the set $S_0\cup
S_3\cup\cdots\cup S_{6k}$, $S_1\cup S_4\cup\cdots\cup S_{6k+1}$, and
$S_2\cup S_5\cup\cdots\cup S_{6k+2}$ respectively.
Here $Y$ is a constant $1/2$, $X$ equals $(1/Y)^k=2^{\theta(\sqrt{\log{n}})}$,
so $X^{6k}=(1/Y)^{6k^2}>2^{-\log n_0}=1/n_0$.
And $c$ is a parameter to be determined later. We shall
show that in this instance, {\em if buyer 1 comes before buyer 2, we
achieve the required lower bound on the revenue, with an appropriate
choice of $c$}.

Below, $u_i(S)$ denotes the utility of buyer $i$ at price $p$ when
she buys the set $S$.
For convenience, we let $T_i$ denote $\cup_{j\in \zeta_i} S_j$, where
$\zeta_i = \{j| 6k+2\geq j\geq i,\ (j-i) \text{ is divisible by }
3\}$. Similarly we define $T'_i$ as  $\cup_{j\in \zeta_i} S_j'$. We
have the following when
$j\in\{1,2\}$ and $i\leq 3k$:
\begin{equation*}
\begin{aligned}
    v_j(T_i)&\geq\sum_{\ell=0}^{k}v_j(S_{i+3\ell})\\
    &=v_j(S_i)(1+Y^3+\cdots+Y^{3k})\\
	&=\left(1-o\left(\frac{1}{X^2}\right)\right)\frac{v_j(S_i)}{1-Y^3}\\
	p|T_i|&\leq p\sum_{\ell=0}^{\infty}\frac{n_0}{X^{i+3\ell}}\\
	&= p|S_i|\sum_{\ell=0}^{\infty}(\frac{1}{X^{3j}})\\
	&= \left(1+o\left(\frac{1}{X^2}\right)\right)p\size{S_i}
\end{aligned}
\end{equation*}
It is also clear that
$v_j(T_i)<\sum_{\ell=0}^{\infty}v_j(S_i)Y^{3\ell}=v_j(S_i)/(1-Y^3)$ and
$p|T_i|>p|S_i|$. So we have $v_j(T_i)=(1\pm
o(1/X^2))v_j(S_i)/(1-Y^3)$ and $p|T_i|=(1\pm o(1/X^2))p|S_i|$.
Similarly, we also have $v_j(T'_i)=(1\pm o(1/X^2))v_j(S'_i)/(1-Y^3)$
and $p\size{T'_i}=(1\pm(1/X^2))p\size{S'_i}$. Use these facts we get that
\begin{equation*}
\begin{aligned}
    v_2(T_{i+1})&=(1\pm o(1/X^2))v_2(S_{i+1})/(1-Y^3)\\
    &=(1\pm o(1/X^2))(v_2(S_{i})-v_2(S'_{i}))/(1-Y^3)\\
    &=(1\pm o(1/X^2))(v_2(T_i)-v_2(T'_i))\\
    &=(1\pm o(1/X^2))v_2(T_i\setminus T'_i)
\end{aligned}
\end{equation*}
Therefore, when the price $p$ is non-trivially large, that is,
$p|T_i\setminus T'_i|>1/X^2\geq v_2(T_i\setminus T'_i)/X^2$,
we have $u_2(T_{i+1})>u_2(T_i\setminus T'_i)$ since both sets have
essentially the same valuation, and the former has significantly fewer
items and hence costs less.
The following facts will be useful for the rest of the proof. There exists
$a_i$, $b_i$, and $c_i$ such that:
\begin{equation*}
\begin{aligned}
	u_1(T'_{i+1})>u_1(T'_i)&\Leftrightarrow p>a_i=
	\frac{(1\pm O(1/X))cX^{i+1}Y^i}{n_0(1+Y+Y^2)}\\
    u_2(T_{i+1}\setminus T'_{i+1})>u_2(T_i)&\Leftrightarrow p>b_i=
	\frac{(1\pm O(1/X))X^iY^i(1+Y)}{n_0(1+Y+Y^2)}\\
	p|T_i|>1/X&\Leftrightarrow p>c_i=\frac{X^{i-1}}{n_0}
\end{aligned}
\end{equation*}
We shall ensure the following
constraints:
\begin{eqnarray}
    u_1(T'_{i+1})>u_1(T'_i) &\Rightarrow& u_2(T_{i+1}\setminus T'_{i+1})>u_2(T_i)\label{eqn-fact-1}\\
    p|T_{i+1}|>1/X &\Rightarrow& u_1(T'_{i+1})>u_1(T'_i)\label{eqn-fact-2}
\end{eqnarray}
Equation \ref{eqn-fact-1} implies that the first buyer prefers
$T'_{i+1}$ than $T'_i$ only if she can ensure that the second buyer
will not buy the set $T_i$ even if $T'_{i+1}$ is taken away.
Equation \ref{eqn-fact-2} indicates that when the price is
non-trivially high such that the set $T_{i+1}$ will give high
revenue, the first buyer will buy $T'_{i+1}$ and prevent the second
buyer from buying $T_i$. Therefore, we shall have $c_{i+1}>a_i>b_i>c_i$ and
thus the parameter $c$ satisfies that:
\begin{equation*}
    \frac{1+Y+Y^2}{X}>c>\frac{Y}{X}
\end{equation*}

Recall that $Y=1/2$, we have $1+Y+Y^2=1.75$. So we let $c=1/X$ and Equation
\ref{eqn-fact-1} and \ref{eqn-fact-2} are guaranteed and we
have $c_{i+1}>a_i>b_i>c_i$. We now consider various possibilities
for the choice of $p$:

\begin{itemize}
\item If the single price $p$ is in the range $[c_i,a_i]$, then
buyer 1 will buy all items in $T'_i$, and buyer 2 will buy all items
in $T_{i+1}$ since $u_2(T_{i+1})>u_2(T_i\setminus T'_i)>u_2(T_{i-1})$.
Therefore, the profit is $p|T_{i+1}|+p|T'_i|<2/X<\opt/(X/2)$. This is
true for all $i<k$.

\item If the single price $p$ is in the range $[a_i,c_{i+1}]$,
then buyer 1 will buy all items in $T'_{i+1}$, and buyer
2 will buy all items in $T_{i+2}$ since
$u_2(T_{i+2})>u_2(T_{i+1}\setminus T'_{i+1})>u_2(T_i)$. So the profit
is $p|T_{i+2}|+p|T'_{i+1}|<2/X<\opt/(X/2)$. This is true for all
$i<k-1$.

\item For $p>c_k$, the only items that can be sold is
$T_k\cup T_{k+1}\cup T_{k+2}$, and the revenue is at most
$v_1(T'_k)+v_2(T_k)\leq 2(Y^k) = 2/X<\opt/(X/2)$.
\end{itemize}

Thus we get an $\Omega(X)=2^{\Omega(\sqrt{\log n})}$ gap between
revenue and the optimal social welfare. Since the total number of
items $n=|\cup_{1\leq i\leq k} S_i|$ is between $n_0$ and $2 n_0$, the
construction and proof of lower bound for the 2 player instance is
complete.

\paragraph{Extensions of the two-player instance:}
We now complete the proof of Theorem \ref{stat-unif-lb} by extending
the above two-player instance..
We present three hard instances such that
\begin{enumerate}
\item {\em Instance 1:} The $2^{\Omega(\sqrt{\log n})}$
lower bound holds even if the buyers come in random order;
\item {\em Instance 2:} The lower bound holds even if the
seller can choose the order of the buyers;
\item {\em Instance 3:} The lower bound still holds if all the buyers
  have identical valuation functions.
\end{enumerate}
The construction of each of the latter instances is based on the
previous instance. Instance 1 is almost identical to one
of the scenarios used in Balcan et al. \cite{BBM08}.

\medskip
\noindent
{\bf Instance 1:} Consider a setting where there are $m$ buyers. One
of them has the same valuation function as buyer 2. The other $m-1$ of
them share the same valuation function as buyer 1. Each of the other
$m-1$ buyers has its own shadow copy of $T'_i$. Then the profit is at
most $1+(m-1)/X$ if the special buyer comes first and at most $m/X$
otherwise. So the expected revenue is $(1/m)(1+(m-1)/X)+(m-1)/m(m/X)<1/m+m/X$.
The optimal social welfare is $\opt>1$.
Let $m=\sqrt{X}$ and we have the $2\sqrt{X}=2^{\Omega{(\sqrt{\log n})}}$
lower bound.

\medskip
\noindent
{\bf Instance 2:} Now consider another instance in which we
replicate (items only, not the buyers) $m$ copies of this
setting such that each buyer is the buyer 2 in exactly 1 replicate.
For each buyer,
there is an additive component for each combination of the additive
components in all $m$ copies.
Suppose the first buyer buys $T_i$ in the replicate in which
she is the buyer 2 and buys $T'_j$ in all other replicate.

\begin{itemize}
    \item If $j\neq i$, then each of the other $m-1$ buyers will buy
        $S_i$ in the replicate in which she is the buyer 2
        and her own shadow copies of $T'_j$. It follows from our choice
        of parameters that $j\neq i$
        implies each copy of the $T_i$ and $T'_j$ provides at most $1/X$
        revenue. So the total revenue is at most $m^2/X$.
    \item If $j=i$, then each of the other $m-1$ buyers will buy
        $T_{i+1}$ in the replicate in which she is buyer 2
        and her own shadow copies of $T'_i$. In this case, the revenue
        obtained from $T_i$ is at most $1$ and each copy of $T'_i$
        and $T_j$ provides at most $1/X$. So the total revenue is
        at most $1+(m^2-1)/X$.
\end{itemize}

It is clear that the optimal social welfare is $\opt>m$, thus giving
a gap of $1/m + m/X=2\sqrt{X}=2^{\Omega(\sqrt{\log{n}})}$.

\medskip
\noindent
{\bf Instance 3:} Finally, consider a setting where all buyers are
identical and share
the same valuation function $v(S)=\max_{1\leq i\leq m}v_i(S)$, where
$\{v_i | 1\leq i\leq m\}$ are the buyer valuation functions as defined
in Instance 2. Each
of the buyer comes and buys some $S_i$ in one of the replicate and
$S'_j$ in the all other replicates. We say the buyer occupies the
replicate which contains $S_i$ in this case.
Note that a copy of $S_i$ is also available in some unoccupied
replicate, and buying the copy in the unoccupied replicate does not
effect the behavior of the buyers who come after her. So for each of
the possible scenario, there is equivalent scenario in which each
buyer occupies an unoccupied replicate. Therefore, we have the same
$2^{\Omega(\sqrt{\log{n}})}$ lower bound as in the previous setting.

This completes the proof of Theorem \ref{stat-unif-lb}.

\section{Dynamic Uniform Pricing Strategies}\label{dynunif-section}

We now present a dynamic uniform pricing strategy that achieves an
$O(\log^2 n)$-approximation to the revenue when buyer valuations are subadditive.
This improves upon the previous best known approximation factor of
$2^{O(\sqrt{\log n \log \log n})}$~\cite{BBM08} for the \IP problem.
Our strategy makes the assumption that the seller knows $\opt$, the maximum
social welfare, to within a constant factor. However, this assumption can easily
be eliminated by using Lemma \ref{param-assumption}, worsening the approximation
ratio of the strategy by a poly-logarithmic factor.

We will also establish an almost matching lower bound result which shows that no
dynamic uniform pricing strategy can achieve $o(\log^2 n/\log \log^2
n)$-approximation even when buyers are restricted to XOS valuations,
the seller knows the value of $\opt$, buyer valuation functions, and
is allowed to specify the order of arrival of the buyers!

\subsection{A Dynamic Uniform Pricing Algorithm}\label{dyn-alg-subsection}

The algorithm follows a simple strategy.
Let $k=\lceil \log n \rceil + 1$, and let $p_i=\opt/2^i$ (recall that
$\opt$ denotes the maximum social welfare). The algorithm starts at
time $0$ by choosing a {\em threshold value}  $p^*$ from the set
$\{p_1,p_2\ldots p_{k+1}\}$, uniformly at random. Upon arrival of any
buyer, the algorithm chooses a price $\hat{p}$ uniformly at random
from the set $\{p_1,p_2 \ldots,p^*\}$, and assigns the price $\hat{p}$
to all items that are yet unsold.

\begin{theorem}\label{thm-dynamic}
If the buyer valuations are subadditive, then the
 expected revenue obtained by the dynamic strategy above is $\Omega(\opt/\log^2
 n)$.
\end{theorem}

The following lemma is key to the proof of Theorem
\ref{thm-dynamic}. It says that if the threshold is ``correctly''
chosen, then our dynamically reset prices give a large fraction of the
maximum possible revenue. 


\begin{lemma}\label{lem-dynamic}
Suppose when the $i^{th}$ buyer $B_i$ arrives, there remains a set $L_i^j$
of unsold items such that $v_i(L_i^j)\geq p_j |L_i^j|$, where $v_i$
is the valuation function of $B_i$. Then if the seller picks a
price from $\{p_1,p_2,\cdots,p_{j+1}\}$ uniformly at random, and
prices all items at this single price, it receives an
expected revenue of at least $p_j |L_i^j|/2(j+1)$ from this buyer.
\end{lemma}

\begin{proofof}{Lemma \ref{lem-dynamic}}
Let $I'$ be the set of unsold items when the buyer $B_i$
arrives.
Since this is a single buyer setting with uniform pricing, Lemma
\ref{decreasing-size} applies. Thus the number of items
sold is a non-increasing function of the price set on all
items, and equation (\ref{netutil-eqn}) is applicable.

Now if the uniform price chosen by the seller is $p_{j+1}$,
then buying the set $L_i^j$ would give $B_i$ a utility of at least $v_i(L_i^j) -
p_{j+1} |L_i^j|\geq p_j |L_i^j|- p_{j+1}|L_i^j|$
since $v_i(L_i^j) \geq p_j |L_i^j|$ by the assumption of the lemma. Thus

 \begin{equation}
 \label{eqn:util_lowerbound}
  \u(B_i,I',p_{j+1}) \geq p_j |L_i^j|- p_{j+1}|L_i^j| = \frac{p_j |L_i^j|}{2}
 \end{equation}

Suppose $q_s > p_{j+1} \geq q_{s+1}$, for some $s\leq l$. Then, since
$\u(B_i,I',q_0)=\u(B_i,I',q_1)=0$, and also that $q_1\leq \opt$, so

 \begin{equation*}
  \begin{aligned}
    \u(B_i,I',p_{j+1}) &=
    \u(B_i,I',q_s)+  \left( \u(B_i,I',p_{j+1})-\u(B_i,I',q_s) \right) \\
    &=\sum_{t=1}^{s-1} \left( \u(B_i,I',q_{t+1})-\u(B_i,I',q_{t}) \right) +
    n_s(q_s - p_{j+1}) \\
    &= \sum_{t=1}^{s-1}n_t(q_t - q_{t+1})+n_s(q_{s}-p_{j+1})
  \end{aligned}
 \end{equation*}

The above sum can be seen as an integral of the following step
function $f$ from $p_{j+1}$ to $q_1$: in the range $[q_{t+1},q_t)$,
$f$ takes the value $n_t$. So we can upper bound it by an upper
integral of $f$. Note that $|f(p)\leq \Phi(B_i,I',p)|\leq |S|$, and
also that $f$ is a decreasing function. Thus we get
\begin{equation*}
  \begin{aligned}
  \u(B_i,I',p_{j+1})   &\leq \sum_{t=0}^j |\Phi(B_i,I',p_{t+1})|(p_t -
  p_{t+1}) \\
    &= \sum_{t=0}^j |\Phi(B_i,I',p_{t+1})|p_{t+1}\\
    & = \sum_{t=1}^{j+1}
    |\Phi(B_i,I,p_{t})|p_{t} \\
  \end{aligned}
 \end{equation*}

Combining with equation~\ref{eqn:util_lowerbound}, we get
 \begin{equation}
  \sum_{t=1}^{j+1} |\Phi(B_i,I,p_{t})|p_{t} \geq \frac{p_j |L_i^j|}{2}
 \end{equation}

Thus the expected revenue obtained from $B_i$ is

$$\sum_{t=1}^{j+1} \frac{|\Phi(B_i,I,p_{t})|p_{t}}{j+1} \geq
\frac{p_j |L_i^j|}{2(j+1)},$$
completing the proof of the lemma.
\end{proofof}

\medskip

\begin{proofof}{Theorem \ref{thm-dynamic}}
 Let $(T_1,T_2,\ldots T_m)$ be an optimal allocation of items to
 buyers $B_1,B_2\ldots B_m$, who has valuation functions
 $v_1,v_2\ldots v_m$ respectively, such that $\sum_{i=1}^m
 v_i(T_i)=\opt$ is the maximum social welfare.
 Also, let $T_i^j$ be the subset of $T_i$ that would be bought by
 $B_i$ if it were shown only the items in $T_i$, and all items were uniformly
 priced at $p_j$. Now consider the case when
 $p^*=p_{j+1}$. Let $R^j$ be the revenue in this case. Let
 ${Z_i^j} \subseteq T_i^j$ be a random variable that
 denotes the subset of items in $T_i^j$
 that are sold before buyer $B_i$ comes. Then
$R^j\geq \sum_{i=1}^m p^*|{Z_i^j}|=\sum_{i=1}^m p_j|{Z_i^j}|/2$

 Note that $v_i(T_i^j \setminus {Z_i^j}) \geq p_j |T_i^j \setminus
 {Z_i^j}|$ by Lemma \ref{unsold-value}.
 So, by Lemma \ref{lem-dynamic}, conditioned on
 the set ${Z_i^j}$, the expected revenue received from $B_i$ is
 at least $\left( p_j|T_i^j\setminus {Z_i^j}| \right)/2(j+1)$. Thus, conditioned on
 the sets ${Z_i^j}$ for all $i$, we have

 \begin{equation*}
   \E{R^j|{Z_i^j} \ \forall 1\leq i\leq m}\geq \Omega\left( \sum_{i=1}^m
   \left( p_j|{Z_i^j}| + \frac{p_j |T_i^j\setminus
   {Z_i^j}|}{j} \right) \right) = \Omega \left( \sum_{i=1}^{m} \frac{p_j|T_i^j|}{j} \right)
 \end{equation*}

 Since the value on the right-hand side above is independent of the variables
 ${Z_i^j}$ on which the expectation of $R^j$ is conditioned on, we get

\begin{equation*}
 \E{R^j} =  \Omega( \sum_{i=1}^{m} \frac{p_j|T_i^j|}{j} )
\end{equation*}

Thus the expected
 revenue $R=\sum_{j=0}^k R^j$ of our dynamic strategy is given by

 \begin{equation}
 \label{eqn:E[R]}
   \begin{aligned}
   \E{R} &=\frac{1}{k+1} \sum_{j=0}^k \E{R^j} = \Omega\left(
   \sum_{j=0}^k \sum_{i=1}^m \frac{p_j |T_i^j|}{k^2} \right) =
   \Omega\left(
   \sum_{i=1}^m \sum_{j=0}^k  \frac{p_j |T_i^j|}{k^2} \right)
   \end{aligned}
 \end{equation}

Since $k=\lceil \log n \rceil$, and $\opt\geq H_{v_i}(T_i)$,
from Lemma \ref{BBM-revsum} and Equation~(\ref{eqn:E[R]}),
it follows that

\begin{equation*}
\sum_{j=0}^k p_j |T_i^j|
\geq \Omega\left( v_i(T_i) - \frac{|T_i|\opt}{2n} \right)
\end{equation*}

 Thus we have


\begin{equation*}
\begin{aligned}
\E{R} &= \Omega\left(  \frac{1}{k^2} \left( \sum_{i=1}^m
    v_i(T_i) - \sum_{i=1}^m \frac{|T_i|\opt}{2n} \right) \right)\\
&= \Omega\left( \frac{1}{k^2} \left( \opt - \frac{\opt}{2} \right)
\right) = \Omega\left( \frac{\opt}{\log^2 n} \right)
\end{aligned}
\end{equation*}

\end{proofof}

\subsection{Lower Bound for Dynamic Uniform
  Pricing}\label{dyn-lb-subsection}

We shall now construct a family of instances of the problem, where the buyers
have distinct XOS valuations, with $O(\log n/\log\log n)$ additive
components in each valuation function, such that no
dynamic uniform strategy can achieve an
$o(\log^2 n/\log \log^2 n)$-approximation, even in the full information setting,
and when the seller can even specify the order in which the buyers
should arrive.

\begin{theorem}\label{dyn-hardness}
There exists a set of buyers with XOS valuations, such that if the
seller is restricted to using a dynamic uniform pricing strategy, then
even when the seller has full information of buyer valuation functions and
can even choose the order of arrival of the buyers,
the revenue produced is $O((\log\log n)^2/\log^2 n)$ times $\opt$, where $n$ is the
number of items.
\end{theorem}
\begin{proof}
Let $B_1,B_2\ldots B_m$ denote the buyers.
Our construction will use three integer parameters $k, F,$ and $Y$, to be specified later.
These parameters will satisfy the conditions that $k>1$, $F>1$, $Y > 4$,
and $m\geq 2Y\geq 4k$.
Let $f(i)=(i+1)F/Y^{i}$. Then, $f(0) > f(1) > \ldots > f(k) >
f(k+1)$.

For each buyer $B_i$, we create $2(k+1)$ disjoint sets of items
$S_{i0},S_{i1}\ldots S_{ik}$ and $S_{i0}',S_{i1}'\ldots S_{ik}'$ such
that $|S_{ij}| = |S_{ij}'| = Y^j$ items each. Let
$S_i=\cup_{0\leq j\leq k} S_{ij}$ and $S_i'=\cup_{0\leq j\leq k}
S_{ij}'$. We call the items in $S_i$ as {\em shared} and those in
$S_i'$ as {\em private}. The private items of $B_i$ are valued by
buyer $i$ only, and has zero value to all other buyers.

The valuation function $v_i$ of buyer $B_i$ is constructed as an XOS
valuation with $(k+2)$ additive functions $v_{i0},v_{i1}\ldots
v_{i(k+1)}$ in its support, that is, $v_i = \max_{0\leq j\leq k+1}
v_{ij}$. For $0\leq j \leq k$, the valuation function $v_{ij}$ has positive value
only for private items, and is defined as

\[ v_{ij}(x) = \left\{ \begin{array}{lll}
                   f(j)    & \mbox{if $x \in S_{ij}'$} \\                        0   & \mbox{otherwise}
                  \end{array}
          \right. \]

The valuation function $v_{i(k+1)}$
has positive values only for shared items:

\[ v_{i(k+1)}(x) = \left\{ \begin{array}{lll}
                   f(j)    & \mbox{if $x \in S_{ij}$ for $0\leq j\leq k$} \\                        f(j+1)   & \mbox{if $x \in S_{\ell j}$ for $1\leq \ell \leq m, \ell \neq i$, and $0\leq j\leq k$}
                  \end{array}
          \right. \]

This completes the description of the instance. Note that
$v_i(S_{ij}')=f(j) |S_{ij}'| = (j+1)F$, and that

\begin{equation*}
v_i(S_i) = \sum_{j=0}^k f(j)
|S_{ij}| = \sum_{j=0}^k (j+1)F = \Omega( k^2F)
\end{equation*}

Thus if we allocate each set $S_i$ to buyer $B_i$ for $i=1,2\ldots m$,
the social welfare obtained is $\Omega( mk^2F)$, and hence $\opt$ is
$\Omega( mk^2F)$.

Consider now the arrival of some buyer $B_i$ at time $t$.
By our construction of the valuation function, $B_i$ will either
buy only shared items or buy only private items, but not both. If the buyer $B_i$ were
to buy shared items, and the price of each item is set at $f(j)\geq
p> f(j+1)$, then $B_i$ would pick up all remaining items in

\begin{equation*}
\left( \bigcup_{0 \le t \le j} S_{it} \right) \bigcup
\left( \bigcup_{1\leq \ell \neq i \leq m} \bigcup_{ 0\leq t\leq (j-1)
  } S_{\ell t} \right)
\end{equation*}

Since $\sum_{0\leq t\leq j} Y^t \leq 2 Y^j$,
the total price that $B_i$ would pay to the seller is bounded by

\begin{equation*}
f(j) (2Y^j + 2mY^{j-1}) = 2(j+1)F + 2m(j+1)F/Y = 2(1 + m/Y)(j+1)F
\end{equation*}

We now consider the maximum revenue generated if $B_i$ were
to buy a subset of its private items. Note that when $B_i$ arrives,
all private items of $B_i$ are still unsold.
Suppose $B_i$ were to buy private items. What is the maximum
revenue we can get? For this, note that if the price of each item is
$(j+1)F/jY^j$, then the utility from buying $S_{ij}'$ is

\begin{equation*}
(j+1)F - (j+1)F/j =(j+1)(j-1)F/j = (j-1/j)F,
\end{equation*}

 and the utility from buying $S_{i(j-1)}'$ is

\begin{equation*}
jF - (j+1)F/jY > (j-1/j^2)F> (j-1/j)F \text{,\ \ \  since $Y>2j$.}
\end{equation*}

 So at this price, the set $S_{i(j-1)}'$ is preferred
by $B_i$ than $S_{ij}'$, and since the items in sets $S_{it}'$ for
$t>j$ have
less value than the price, they are not even considered. For a greater
price, the utility of $S_{i(j-1)}'$ must continue to dominate that
of $S_{ij}'$, since the former has fewer items. So at most
$Y^{j-1}$ items are bought when the price is at least $(j+1)F/jY^j$,
for all $j\geq 1$. This implies that the revenue obtained from $B_i$
when she buys from her private items is at most $Y^j ((j+1)F/jY^j) <
2F$.

Consider any ordering of buyers. If the price is ever set at more than
$f(0)$, then no item is sold in that round, while if the price set is
$f(k+1)$ or lower, all items are sold in that round and the revenue
generated is at most $2mY^k f(k+1)=2m(k+2)/Y$.
Consider the first time when the price set in a round is
at most $f(j)$ but greater than $f(j-1)$, for some $0\leq j\leq k$. We
call this round a {\em $j$-good sale}, and let $B_i$ be the buyer. In
a $j$-good sale, $B_i$ may buy all remaining items in $S_{it}$ for all
$0\leq t\leq j$, plus all items in $S_{lt}$ for all $1\leq \ell \leq m,
\ell \neq i$ and $0\leq t\leq j-1$, thus giving a revenue of at most

\begin{equation*}
2(1+ m/Y)(j+1)F \leq O((m/Y)kF)
\end{equation*}

However, consider any time when a price in the range
$(f(j-1),f(j)]$ appears again, and let $B_l$, $\ell \neq i$ be the buyer
who faces this price. If $B_l$ were to buy shared items, the only
items that are valued higher than the price and still remaining are
those in $S_{\ell j}$, since $B_i$ took away whatever was remaining of
$S_{\ell t}$ for all $t<j$. Note that the only reason the shared items
could have given $B_i$ a better utility was that the shared items
had additive valuation, while the private items had XOS valuation, so
she got no benefit in picking up multiple sets of private
items. However, since only one feasible set $S_{\ell j}$ of the shared
items is left, this advantage has vanished, and the revenue from $B_{\ell}$
is the same as the revenue if there were no shared items at all. As
discussed above, the revenue from $B_{\ell}$ in this case is at most $2F$.

Finally, since a $j$-good sale can happen
at most once for any $1 \le j \le (k+1)$, the
total revenue generated fro all $j$-good sales is $O\left( (mk^2 F)/Y \right)$.
The remaining rounds each give a revenue of at most $2F$,
contributing in total $O(mF)$ to the revenue. Thus
the revenue obtained by any dynamic uniform strategy, {\em for any
ordering of buyers}, is
$O\left((1 + k^2/Y) mF \right)$.

Now since the maximum social welfare is $\Omega(k^2 mF)$, the approximation
factor achieved is bounded from below by $\Omega\left( (k^2 Y)/(k^2 + Y) \right)$.
For any $k> 10$, if we set $Y=k^2$ and $m=2Y$, then $n=\Theta(Y^{k+1}) =
k^{\Theta(k)}$, and the approximation factor is $\Omega(k^2)$.
As $k=\Theta(\log n/\log \log n)$, we get that the
smallest approximation factor that can be achieved is
$\Omega\left((\log n/\log \log n)^2\right)$.
\end{proof}

\section{Dynamic Monotone Uniform Pricing Strategies}\label{dynmon-section}

We now present a simple strategy that uses a monotonically decreasing
uniform pricing for the items. When the number of buyers $m$ is at least
$2\log n$, the strategy gives an $O(\log^2 n)$-approximation to the
revenue provided the buyers arrive in a uniformly random order, that
is, all permutations of the buyers are equally likely to be the arrival
order. As a corollary of this result, we conclude that if the buyers
are identical, no matter the order in which they arrive, this pricing
scheme gives an $O(\log^2 n)$-approximation. The strategy assumes that
the seller knows the number of buyers $m$ (and also $\opt$), and is
{\em deterministic}. Knowing estimates of $m$ and $\opt$ up to constant
factors are also sufficient for the performance of our strategy.

 Let $k=\log n + 1$, and let $\gamma =
2^{ \frac{k}{m}} \geq 1$. Thus $\gamma^m >
2n$. The strategy gives a good guarantee only when $m\geq \log n
+1$. The strategy is as follows: When the $t^{\rm th}$ buyer arrives,
the seller prices all unsold items uniformly at

\begin{equation*}
p[t] = \frac{\opt}{2\gamma^{t}}\ \ \ .
\end{equation*}

Thus the price decreases with time. For
 $m=\omega(\log n)$, the relative decrease in the price for
 consecutive buyers is

\begin{equation*}
\frac{p[t] - p[t+1]}{p[t]} =
 \left(1- \frac{1}{\gamma}\right) =
 \Theta(\frac{\log n}{m}),
\end{equation*}

 which tends to zero, and so the price
 decreases smoothly with time.


\begin{theorem}\label{thm-dynmon}
Suppose $m\geq \log n +1$, and suppose that the buyer valuations are
subadditive.  If the ordering of buyers in which they arrive is
uniformly random (that is, all permutations are equally likely), then
the expected revenue of the dynamic monotone uniform pricing scheme
above is $\Omega( \frac{\opt}{\log^2 n} )$.
\end{theorem}
\begin{proof}
 Let $(T_1,T_2,\ldots T_m)$ be an optimal allocation of items to
 buyers $B_1,B_2\ldots B_m$, who has valuation functions
 $v_1,v_2\ldots v_m$ respectively, such that $\sum_{i=1}^m
 v_i(T_i)=\opt$ is the maximum social welfare.
 Also, let $T_i^j$ be the subset of $T_i$ that would be bought by
 $B_i$ if it were shown only the items in $T_i$, and all items were uniformly
 priced at $\opt/\gamma^j = 2 p[j]$.

Fix a buyer $B_i$. Let $R_i$ be a random variable that denotes the
revenue obtained by the seller from $B_i$. Let $R_i'$ be a random
variable that denotes the revenue obtained by the seller by selling
items in $T_i$. Then, if $R$ is a random variable that denotes the
total revenue obtained by our strategy, we have $R=\sum_{i=1}^m R_i$
and $R\geq \sum_{i=1}^m R_i'$, so $R\geq \sum_{i=1}^m \frac{R_i +
  R_i'}{2}$.

Fix a permutation $\pi$ of all buyers except $B_i$. We shall say that
the event $\pi$ occurs if these buyers arrive in the relative order
given by $\pi$, with $B_i$ arriving somewhere in between. We shall now
compute $\E{R_i + R_i'|\pi}$.

Let $\pi_j$ denote the permutation of all the buyers formed by
inserting $B_i$ after the $(j-1)^{th}$ but before the $j^{th}$
position in $\pi$, whichever exists, for $1\leq j\leq m$. That is
$B_i$ comes in as the $j^{th}$ buyer in $\pi_j$. Let $Z_i^j$ denote
the number of items that were sold before the arrival of $B_i$ when
the arrival sequence of buyers is $\pi_j$. Note that $Z_i^j$ is no
longer a random variable once $\pi_j$ is fixed, and neither are $R_i$
and $R_i'$. Also note that $\pr{\pi_j | \pi} = 1/m$. Thus,

\begin{equation}
\label{mon-eqn1}
\E{R_i' | \pi} \geq \frac{1}{m} \sum_{j=1}^{m} p[j-1] |Z_i^j| \geq
\frac{1}{m} \sum_{j=1}^{m} p[j] |Z_i^j|
\end{equation}

Let $S_i^j$ be the set of items bought by $B_i$ when the permutation
of buyers is $\pi_j$, let $\u_i^j$ be the utility derived by $B_i$
in this process, and let $R_i^j$ be the revenue obtained from $B_i$ in
the process. For $1\leq j<m$, note that when the permutation is $\pi_j$,
then when $B_i$ arrives, the set $S_i^{j+1}$ is also available, and
$B_i$ prefers $S_i^j$ over this set at price $p[j]$. Thus

\begin{equation*}
\begin{aligned}
\u_i^j &= v_i (S_i^j) - p[j] |S_i^j|\\
&\geq v_i(S_i^{j+1}) - p[j] |S_i^{j+1}|= v_i(S_i^{j+1}) - p[j+1]
|S_i^{j+1}| - (\gamma - 1) p[j+1] |S_i^{j+1}|\\
&= \u_i^{j+1} - (\gamma - 1) R_i^{j+1}
\end{aligned}
\end{equation*}

This implies that $R_i^{j+1} \geq \frac{1}{\gamma -1} (\u_i^{j+1} -
\u_i^j)$, for $1\leq j<m$. Also note that $v_i(S_i^1) - p[0] |S_i^1|
\leq 0$, since no bundle of items can have value greater than
$p[0]=\opt$. So $v_i(S_i^1) - p[1] |S_i^1| - (\gamma-1)p[1] |S_i^1| =
\u_i^1 - (\gamma-1) R_i^1\leq 0$, or $R_i^1 \geq \frac{1}{\gamma-1}
\u_i^1$. Thus, adding the terms $R_i^t$, we find that the terms
telescope, and

\begin{equation*}
\sum_{t=1}^j R_i^t \geq \frac{1}{\gamma-1} \left( \sum_{t=2}^j
  \left(\u_i^{j} - \u_i^{j-1} \right)  + \u_i^1 \right) =
\frac{1}{\gamma-1} \u_i^j
\end{equation*}

By Lemma \ref{unsold-value}, we have $v_i(T_i^{j}\setminus Z_i^j)
\geq 2p[j] |T_i^{j}\setminus Z_i^j|$. So the utility of
$T_i^{j}\setminus Z_i^j$ to buyer $B_i$ at price $p[j]$ is $(2p[j]
- p[j])|T_i^{j}\setminus Z_i^j| = p[j]|T_i^j\setminus
Z_i^j|$, which is at most $\u_i^j$. Thus

\begin{equation*}
\sum_{t=1}^j R_i^t \geq \frac{1}{\gamma-1}\u_i^j =
\frac{p[j]|T_i^j\setminus Z_i^j|}{\gamma -1}
\end{equation*}

Using the above equation, we get

\begin{equation*}
\begin{aligned}
m \sum_{j=1}^m R_i^j &\geq \sum_{j=1}^m \sum_{t=1}^j R_i^t =
\frac{1}{\gamma-1}\sum_{j=1}^m p[j] |T_i^j\setminus Z_i^j|\\
\Rightarrow \sum_{j=1}^m R_i^j &\geq \sum_{j=1}^m \frac{p[j]
|T_i^j\setminus Z_i^j|}{(\gamma - 1) m} \geq \Omega \left(\frac{p[j]
|T_i^j\setminus Z_i^j|}{\log n}\right)\\
\end{aligned}
\end{equation*}

The last inequality follows from the fact that $\gamma - 1 =
\Theta\left(\frac{\log n}{m}\right)$.

Note that $\E{R_i|\pi} = \frac{1}{m}\sum_{j=1}^m R_i^j$. Combining
with equation (\ref{mon-eqn1}), we get that

\begin{equation*}
\begin{aligned}
\E{R_i + R_i' | \pi} &= \frac{1}{m} \Omega \left( \sum_{j=1}^m p[j]
  \left(\frac{|T_i^j\setminus Z_i^j|}{\log n} + |Z_i^j|\right)
\right)\\
&\geq \frac{1}{m} \Omega \left( \sum_{j=1}^m \frac{p[j]|T_i^j|}{\log n}
\right) = \Omega \left(\frac{1}{m\log n}  \left( \sum_{j=1}^m
  \frac{\opt}{\gamma^j}|T_i^j| \right) \right)
\end{aligned}
\end{equation*}

Using Lemma \ref{BBM-revsum}, we get that

$$\sum_{j=1}^m
  \frac{\opt}{\gamma^j}|T_i^j| \geq \frac{1}{\gamma -1} \left(
    v_i(T_i) - \frac{|T_i|\opt}{\gamma^m} \right) \geq \frac{1}{\gamma -1}
  \left(  v_i(T_i) - \frac{|T_i|\opt}{2n}\right)$$

Again using the fact that $\gamma - 1 = \Theta\left( \frac{\log n}{m}
\right)$, we get that

\begin{equation*}
\E{R_i + R_i' | \pi} \geq \Omega \left( \frac{1}{\log^2 n} \left(
    v_i(T_i) - \frac{|T_i|\opt}{2n}\right)\right)
\end{equation*}

Since the right-hand-side of the above equation is independent of
$\pi$, we conclude that $\E{R_i + R_i'} \geq \Omega \left(
  \frac{1}{\log^2 n} \left( v_i(T_i) - \frac{|T_i|\opt}{2n} \right)
\right)$. Thus we get that the expected revenue is

\begin{equation*}
\begin{aligned}
\E{R} &= \frac{1}{2}\E{R_i + R_i'} \geq \Omega \left(\frac{1}{\log^2 n}
  \left( \sum_{i=1}^m v_i(T_i) - \sum_{i=1}^m \frac{|T_i|\opt}{2n}
  \right) \right)  \\
&= \Omega\left( \frac{1}{\log^2 n} \left( \opt -
    \frac{\opt}{2} \right) \right) = \Omega\left( \frac{\opt}{\log^2
    n} \right)
\end{aligned}
\end{equation*}
\end{proof}

\section{Static Non-Uniform Pricing}\label{statunif-section}

Another approach to get around the weak performance barrier for static uniform pricing,
is to consider static {\em non-uniform} pricing, which allows the seller to post different
prices for different items but the prices remain unchanged over time.
In Section \ref{stat-lb-section} we showed
that there exist instances with identical buyers
where no static uniform pricing can achieve better than
$2^{\Omega( \sqrt{\log{n}})}$-approximation even in the
full information setting.
 Surprisingly, this hardness result breaks
down if we consider non-uniform pricing using only two distinct
prices.

\subsection{Full Information Setting}
We first introduce the $(p,\infty)$-strategies, i.e. the seller
posts price $p$ for a subset of the items and posts $\infty$ for all
other items. The intuition is by using this strategy the seller can
prevent the buyers from buying certain items (high utility but low
revenue) and thus achieve better revenue.
The proof of the theorem below depends on the performance of the
following dynamic monotone strategy.
Let $k=\ceil{\log n} +1$ and $m'=\floor{m/(k+1)}$. Recall that $p_i=\opt/2^i$
for $i=1,2,\cdots,k$. The seller posts a single price $p_1$ for the
first $m'$ buyers, then she posts a single price $p_2$ for the next
$m'$ buyers, and so on and so forth. We call each time period that
the seller posts a fixed price a {\em phase}, and we call this
strategy the $k$-phase monotone uniform strategy.
The proof of Theorem \ref{thm-dynmon} can be easily modified to show
that this strategy gives $O(\log^2 n)$-approximation as well.

\begin{theorem}\label{thm-p-inf}
    In the full information setting, if $m\geq \log n +1$, and all
    buyers share the same subadditive valuation function,
    then there exists a $(p,\infty)$-strategy which
    obtains revenue at least $\Omega(\opt/\log^3{n})$.
\end{theorem}
\begin{proof}
	Given that the $k$-phase dynamic monotone uniform strategy
	for identical buyers obtains revenue at least
    $\Omega(\opt/\log^2{n})$, at least one of the $k=\ceil{\log{n} +1}$
    phases contributes $1/k$ fraction of the revenue. Without loss
    of generality, assume the $i^{th}$ phase contributes at least
    $\Omega(\opt/k\log^2{n})=\Omega(\opt/\log^3{n})$ revenue.
    Recall that $m'=\floor{m/(k+1)}$. Suppose
    $T$ is the set of items unsold at the beginning of phase $i$.

    Consider the following $(p,\infty)$-strategy. The seller posts price
    $p=p_{i+1}$ for each item in $T$, and posts $\infty$ for all other
    items. Then when the first $m'$ buyers come, they will behave the
    same as the $m'$ buyers in phase $i$ in the dynamic strategy scenario.
    So the revenue collected is at least $\Omega(\opt/\log^3{n})$.
\end{proof}

\subsection{Buyers with $\ell$-XOS Valuations}

The above theorem shows a clear gap between the power of uniform
pricing and the power of non-uniform pricing in the full information
setting. However, it crucially uses the knowledge of the valuation
function and the fact that all buyers are identical;
information that is usually not known to the seller.
Hence strategies in the limited information setting are more
desirable in practice.

Fortunately, we find that considering static non-uniform pricing is
also beneficial in the limited information setting. We first note that
if the buyer order is randomized, then it is quite easy to get an
$O\left(m\log n\log \opt(\log\log \opt)^2\right)$ approximation using
static uniform pricing, even with general
valuations, and without the assumption of knowing $\opt$. This can be
done as follows: Just focus on selling items to the first buyer. If
$B_i$ is the first buyer, and the algorithm knew the value $v_i(T_i)$,
then using the single buyer (unlimited supply setting) algorithm in
\cite{BBM08}, the algorithm gets $\Omega(v_i(T_i)/\log n)$ in
expectation from the first buyer, and we do not care what it gets from
the other buyers. Thus the expected revenue of the algorithm is
$\frac{1}{m}\left( \frac{\sum_{i=1}^m v_i(T_i)}{\log n}\right) =
\frac{\opt}{m\log n}$. This algorithm
would have to guess $v_i(T_i)\leq \opt$ of the first buyer $B_i$, up
to a constant factor, and can do so by incurring an additional factor
of $O(\log \opt (\log\log \opt)^2)$ as described in Lemma
\ref{param-assumption}.

However, if we require a strategy to give guarantees on expected
revenue against any order of buyers, and in particular an adversarial
ordering, then static uniform pricing cannot give a better bound than
$2^{\Omega(\sqrt{\log n})}$ even when there are only two buyers, with
$3$-XOS valuations. This
is evident from the proof of Theorem \ref{stat-unif-lb}. We now
show a static non-uniform strategy which achieves
polylogarithmic approximation if we assume the valuation functions are
$\ell$-XOS functions where $\ell$  is quasi-polynomial in $n$ and the
number of buyers is polylogarithmic, for all ordering of buyers.

Let $k=\ceil{2\log{n}}$. With probability half, the seller assigns a
single price $p$ randomly drawn from $\{p_1,p_2,\cdots,p_k\}$ to
all items. With probability half, the seller assigns
one of $p_1,p_2,\cdots,p_{k+1}$ uniformly at random for each item. The
price assignment remains unchanged over time.

\begin{theorem}
    \label{thm-non-uniform}
	For $m$ buyers with $\ell$-XOS valuations functions,
	the expected revenue of the above strategy is
        $\Omega\left(\frac{\opt}{m\log\ell\log^3 n}\right)$.
\end{theorem}

Suppose the XOS valuation function of the $i^{th}$ buyer is
$v_i(S)=\max_{1\leq j\leq \ell} a_{i,j}(S)$. For each $m$-tuple
$z=(z_1,z_2,\cdots,z_m)\in[\ell]^m$, define $a_z$ to be an additive
function such that for each item $g\in I$, $a_z(g)=\max_{1\leq i\leq
m} a_{i,z_i}(g)$.
For each $z$ and each $1\leq i\leq k$, let $\Gamma_{z,i}$ denote
the set of items $g$ such $a_z(g)\in[p_i,p_{i-1})$. We say such a
set $\Gamma_{z,i}$ is large if its size is at least $16m\log\ell$
and we say it is small otherwise. Define $A_z$ and $B_z$ as follows:
\begin{equation*}
	A_z=\bigcup_{\substack{\Gamma_{z,i} \geq 16m\log \ell \\
            1\leq i\leq k}} \Gamma_{z,i}\enspace, \ \ \ \ \
	B_z=\bigcup_{\substack{\Gamma_{z,i}<16m\log \ell \\ 1\leq
            i\leq k}} \Gamma_{z,i}\enspace.
\end{equation*}

In the case where the seller posts one of $p_1,p_2,\cdots,p_{k+1}$
uniformly at random for each item, let $\Pi_i$ denote the set of items
which are priced $p_i/2$.

The following two lemmas are crucial to the proof of Theorem
\ref{thm-non-uniform}.

\begin{lemma}
    \label{lemma-non-uniform-1}
    If the seller posts a single price $p$ randomly drawn from
    $\{p_1,p_2,\cdots,p_k\}$ for all items, then the expected revenue is at
    least $\Omega(a_z(B_z)/m\log\ell\log n)$ for any $z\in[\ell]^m$.
\end{lemma}
\begin{proof}
    Let $R_i$ denote the revenue if the seller posts a single price
    $p_i$. When the seller posts a single price $p_i$ for all items, the
    buyers will buy at least one item if
    $B_z\cap\Gamma_{z,i}$ is not empty. Note that
    $\size{B_z\cap\Gamma_{z,i}}<m\log\ell$, we have $R_i\geq
    a_z(B_z\cap\Gamma_{z,i})/m\log\ell$. Since $k=\ceil{2\log n}$,
   	the expected revenue is
    at least
    \begin{equation*}
        \frac{1}{k}\sum_{i=1}^k R_i\geq
        \sum_{i=1}^k\frac{a_z(B_z\cap\Gamma_{z,i})}{km\log\ell}=
        \frac{a_z(B_z)}{km\log\ell}=
		\Omega\left(\frac{a_z(B_z)}{m\log\ell\log n}\right)\enspace.
    \end{equation*}
\end{proof}

\begin{lemma}
    \label{lemma-non-uniform-2}
    If the seller posts one of $p_1,p_2,\cdots,p_{k+1}$ uniformly
    at random for each item,
	then with probability at least $3/4$ we have for every
    $z\in[\ell]^m$,
    $\size{\Pi_i\cap\Gamma_{z,i}} \geq
    \size{A_z\cap\Gamma_{z,i}}/2k$.
\end{lemma}
\begin{proof}
    If $\Gamma_{z,i}$ is small then $\size{A_z\cap\Gamma_{z,i}}=0$
    and the given equation is trivially true. Now suppose
    $\Gamma_{z,i}$ is large, that is, $\size{\Gamma_{z,i}}\geq16m\log\ell$.
	Note that each item in $\Gamma_{z,i}$ has
    probability $1/k$ of being priced $p_i/2$. Using Chernoff bounds
    and we get that the probability that less than $1/2k$ fraction of
    $\Gamma_{z,i}$ are priced $p_i/2$ is at most $1/2^{2m\log\ell}=1/\ell^{2m}$.
	There are at most $\ell^m$ distinct $m$-tuples $z$. For each $z$
    there are at most $k=\ceil{2\log n}$ sets $\Gamma_{z,i}$. So the
	total number of different $\Gamma_{z,i}$ is at most $\ell^mk<\ell^{2m}/4$.
	By using union bound we finish the proof of this lemma.
\end{proof}

\medskip

We can now complete the proof of Theorem \ref{thm-non-uniform}.

\medskip

\begin{proofof}{Theorem \ref{thm-non-uniform}}
    If there exists some vector $z$ such that
    $a_z(B_z)\geq\opt/320\log^2 n$ then we know from Lemma
    \ref{lemma-non-uniform-1} that the expected revenue is at least
    $\Omega(\opt/m\log\ell\log n)$. Now let us assume
    $a_z(B_z)<\opt/320\log^2 n$ for any $z$.

    By Lemma \ref{lemma-non-uniform-2}, it suffices to prove that
	the expected revenue is high if for each $z\in[\ell]^m$
    $\size{\Pi_i\cap\Gamma_{z,i}}\geq\size{A_z\cap\Gamma_{z,i}}/2k$.
    Suppose $T=(T_1,T_2,\cdots,T_m)$ is the allocation that
	maximizes the social
	welfare, then $\opt=\sum_{i=1}^m v_i(T_i)$. There exists
	$m$-tuple $z'\in[\ell]^m$ such that $a_{i,z'_i}(T_i)=v_i(T_i)$ and thus

    $$\opt=\sum_{i=1}^m a_{i,z'_i}(T_i)\leq
    a_{z'}(I)=a_{z'}(A_{z'})+a_{z'}(B_{z'})\ \ .$$

By our assumption $a_{z'}(A_{z'})\geq\opt-\opt/320\log^2 n\geq\opt/2$
    and hence $a_{z'}(A_{z'}\cap\Gamma_{z',j})\geq \opt/2k$ for
    some $j\in[k]$. Let $Z$ denote the set $\Pi_j\cap\Gamma_{z',j}$ and
    we have $\size{Z}\geq\size{\Gamma_{z',j}}/2k$.
	Since $k=\ceil{2\log n}$, we have
    \begin{equation*}
        p_j\size{Z}\geq \frac{p_{j}\size{\Gamma_{z',j}}}{2k}\geq
        \frac{a_{z'}(\Gamma_{z',j})}{4k}\geq \frac{\opt}{8k^2}
		\geq\frac{\opt}{40\log^2 n}\enspace.
    \end{equation*}

    Suppose the $i^{th}$ buyer buys the set $S_i$ for $1\leq i\leq
    m$ and let $S$ denote the union of all $S_i$. If $\size{S\cap
    Z}\geq\size{Z}/2$ then the revenue is at least $(p_j/2)\size{S\cap Z}
    \geq (p_j/2)(\size{Z}/2)=\Omega(\opt/\log^2 n)$.
    Otherwise, $\size{Z\setminus S}\geq\size{Z}/2$.
	Let $u_i(S_i)$ denote the utility of set $S_i$ to the $i^{th}$
	buyer. We have
    \begin{equation*}
        \sum_{i=1}^mu_i(S_i)\geq\sum_{i=1}^m u_i(Z\setminus S)
        \geq a_{z'}(Z\setminus S)-\frac{p_j}{2}\size{Z\setminus S}
        \geq \frac{p_j}{2}\size{Z\setminus S}\geq
		\frac{\opt}{160\log^2 n}\enspace.
    \end{equation*}
    Hence $\sum_{i=1}^m v_i(S_i)\geq\sum_{i=1}^m u_i(S_i)=\Omega(\opt/\log^2 n)$.
	Note that there exists an $m$-tuple $z''\in[l]^m$ such that
	$a_{i,z''_i}(S_i)=v_i(S_i)$. So
    \begin{equation*}
        a_{z''}(B_{z''})+a_{z''}(A_{z''})=a_{z''}(I)\geq\sum_{i=1}^m
        a_{i,z''_i}(S_i)\geq\frac{\opt}{160\log^2 n}\enspace.
    \end{equation*}
    By our assumption, $a_{z''}(B_{z''})<\opt/320\log^2 n$, so
    $a_{z''}(A_{z''})=\Omega(\opt/\log^2 n)$. Note that an item $g$ is bought
	if and only if its price is less than $a_{i,z''_i}(g)$ for some $i$.
	So all items in
    $\Gamma_{z'',i}\cap\Pi_i$ are bought and with high probability
	the revenue is at least
    \begin{equation*}
    \begin{aligned}
        \sum_{i=1}^k\frac{p_i}{2}\size{\Gamma_{z'',i}\cap\Pi_i}
        &\geq\sum_{i=1}^k\frac{p_i\size{\Gamma_{z'',i}\cap
        A_{z''}}}{4k} = \Omega\left(\frac{a_{z''}(A_{z''})}{k}\right)
        =\Omega\left(\frac{\opt}{\log^3 n}\right)\enspace.
    \end{aligned}
    \end{equation*}

    Hence the proof is complete, the expected revenue is at least
    $\Omega(\opt/m\log\ell\log^3 n)$.

\end{proofof}


\bibliographystyle{abbrv}
\bibliography{pricingbiblio} 


\end{document}